\UseRawInputEncoding
\documentclass{article}

\usepackage{stmaryrd}
\usepackage{lmodern}
\usepackage{diagbox}
\usepackage{mathtools}
\usepackage{amssymb,amsmath}
\usepackage{graphicx}

\usepackage{mathrsfs}

\usepackage{color}
\usepackage{amsfonts}
\usepackage{epsfig}
\usepackage{graphics}

\usepackage{lscape}
\usepackage{latexsym}

\def\proof{\noindent {\bf Proof }}
\def\alp{{\alpha}}

\def\qed{~~\vrule height8pt width4pt depth0pt}
\def\ex{{\mathbb E}}

\def\F2{{\mathbb F}_2}

\def\eps{{\varepsilon}}

\def\Mcal{{\cal M}}
\def\Ecal{{\cal E}}

\def\RS{{\cal RS}}

\newtheorem{conj}{Conjecture}
\newtheorem{prop}{Proposition}

\newtheorem{lemma}{Lemma}
\newtheorem{thm}{Theorem}
\newtheorem{cor}{Corollary}

\newcommand{\f}{{\mathbb F}_{q}}
\newcommand{\fq}{{\mathbb F}_{q}}

\newcommand{\fqx}{{\mathbb F}_{q}^{*}}

\newcommand{\Ex}{{\mathbb E}}
\newcommand{\prob}{{\mathbb P}}
\begin{document}

\title {Improved error bounds for the distance distribution of Reed-Solomon codes
}
 \author{
Zhicheng Gao\\
School of Mathematics and Statistics\\
Carleton University\\
Ottawa, Ontario\\
Canada K1S5B6\\
Email:~zgao@math.carleton.ca  \\
Jiyou Li\\
School of Mathematical Sciences\\
Shanghai Jiao Tong University\\
Shanghai, China, 200240\\
Email:~lijiyou@sjtu.edu.cn\\
 }
\maketitle

\begin{abstract}

We use the generating function approach to derive simple expressions for the factorial moments of the distance distribution over Reed-Solomon codes. We obtain better upper bounds for the error term of a counting formula  given by Li and Wan, which gives nontrivial estimates on the number of  polynomials over finite fields with  prescribed leading  coefficients and a given number of linear factors.
This improvement leads to new results on the classification of deep holes of Reed Solomon codes.
\end{abstract}

\section{ Introduction}

The theory of error-correcting codes plays an important role in modern communications. One particular and important class is the Reed-Solomon code, which is  defined by polynomials evaluations on a given subset in a finite field.   Since the invention by Reed and Solomon in 1960s, Reed-Solomon codes are by far the most fundamental examples of evaluation codes and  the most widely studied linear codes. 

Let $\fq$ be the finite field with $q$ elements of characteristic $p$.
Let $D=\{x_1,\cdots, x_n\}$ be a subset in $\fq$ of
cardinality $n$. For $1\leq k\leq n$, the Reed-Solomon code
$\RS_{n,k}$ has the codewords of the form
$$(f(x_1), \cdots, f(x_n))\in \fq^n,$$
where $f$ runs over all polynomials in $\fq[x]$ of degree bounded by $k-1$.
 The (Hamming) distance  between two words $u,v$, denoted by $d(u, v)$,  is the number of non-zero entries in $u-v$.
For the Reed-Solomon code,
the distance between two codewords represented by polynomials $f$ and $g$ is $n$ minus the number of distinct roots of $f-g$ in $D$, or equivalently the number of distinct
linear factors $x-\alp$ with $\alp\in D$.
The minimum distance of the Reed-Solomon code is $n-k+1$
since  a non-zero polynomial of degree no more than $k-1$ has at most
$k-1$ zeroes. For  $u=(u_1,u_2,\cdots,u_{n})\in \fq^n$,  one associates a unique polynomial $u(x)\in \fq[x]$ of degree
at most $n-1$ such that
$u(x_i)=u_i,$  given by the Lagrange interpolation formula
$$u(x) = \sum_{i=1}^n u_i \prod_{j\not=i}\frac{x-x_j}{x_i-x_j}.$$
Define $\deg(u)$ to be the degree of the associated polynomial $u(x)$
of $u$. Note that  $u$ is a codeword if and only if
$\deg(u)\leq k-1$.

For a given word  $u\in \fq^n$,  define its distance to the code by
$$d(u, \RS_{n,k}): = \min \{d(u,v):v\in \RS_{n,k}\}.$$
The maximum likelihood decoding of $u$ is to find a codeword $v\in
\RS_{n,k}$ such that $d(u,v) = d(u, \RS_{n,k})$. Thus, the computation of
$d(u,\RS_{n,k})$ is essentially the decision version for the maximum
likelihood decoding problem, which is ${\bf NP}$-complete for
general subset $D\subset \fq$. For standard Reed-Solomon code
with $D=\fq^*$ or $\fq$, the complexity of the maximum
likelihood decoding is unknown to be {\bf NP}-complete. This is an
important open problem. It has been shown by Cheng-Wan
\cite{CW04,CW07} to be at least as hard as the discrete logarithm
problem in a large finite  extension of $\fq$.

When $\deg(u)\leq k-1$,  then $u$ is a codeword and thus
$d(u,\RS_{n,k})=0$. We shall assume that $k \leq \deg(u)\leq n-1$. The
following simple result given by Li and Wan gives an elementary bound for $d(u,
\RS_{n,k})$.

\begin{thm}\label{thm1} Let $u\in \fq^n$ be a word such that $k \leq \deg(u)\leq n-1$.
Then,
$$n-k \geq d(u, \RS_{n,k}) \geq n-\deg(u).$$
\end{thm}

\subsection{The deep hole conjecture}

A received word  $u$  is called a deep hole if $d(u,\RS_{n,k})=n-k$, that is, $u$ realizes the covering radius.  When $d(u)=k$, the upper bound in Theorem \ref{thm1}
agrees with the lower bound and thus $u$ must be a deep hole. This
gives $(q-1)q^k$ deep holes. For a general Reed-Solomon code
$\RS_{n,k}$, it is very  difficult to determine if a given word $u$
is a deep hole. In the special case that $d(u)=k+1$, the deep hole
problem is equivalent to the subset sum problem over $\fq$
which is {\bf NP}-complete if $p>2$.

For the standard Reed-Solomon code, that is, $D=\fq$ and
 thus $n=q$, there is the following interesting conjecture of
Cheng-Murray \cite{ChenMu07}.
\begin{conj} (Cheng-Murray) For the standard Reed-Solomon code
with $D=\fq$,  the set $\{u\in \fq^n \big| d(u)=k\}$
gives the set of all deep holes.
\end{conj}
\textbf{Remark} Note that this conjecture is  different from the original one, where $D=\fq^*$ and $q=p$ is prime.  But the exceptions are known. See  examples given by Wu and Hong\cite{WuHong12}; and Zhang, Fu and Liao \cite{ZhangFuLiao13}.

Given a monic polynomial $f(x)\in \f [x]$ of degree $k+\ell$ and an integer $0\leq r\leq k+\ell$, let  $N(f(x), r)$ be  the number  of polynomials $g(x)\in\f[x]$  with  $\deg g(x) \leq k-1$ such that $f(x)+g(x)$ has exactly $r$ distinct roots in $\f$.
Then the deep hole conjecture is equivalent to the fact that the inequality
\begin{align*}
N(f((x),r)>0, \hbox{ for some $r>k$}
\end{align*}
holds for any $f(x)$ of degree $k+\ell$ and any  $\ell\ge 1$.
In particular,  studying when the special case $N(f(x),k+1)>0$ holds for a wider range of parameters will support the deep hole conjecture.

Using the Lang-Weil bound, Cheng and Murray  \cite{ChenMu07} proved that their conjecture is true if $k=O(q^{\frac 17})$ and $l=O(q^{\frac {3}{13}})$.
There are a series of results on this problem along this direction.
 Li and Wan \cite{LiYWan08} proved that the conjecture is true if $k\le q^{\frac 12}$ and $\ell=o(q^{\frac 12})$.
   A different proof was given by Cafure, Matera and Privitelli \cite{Cafure12} with tools from algebraic geometry. A refined result was given by  Liao \cite{Liao11}.  Consequently, Zhu and Wan \cite{ZhuWan12} proved that there
 are constants $c_1, c_2$ such that the  conjecture holds for $k=c_1q$ and $\ell=c_2q^{\frac 12}$.

 Using a different approach, including tools from finite geometry and additive combinatorics,  Zhuang, Li and Cheng proved that
 the deep hole conjecture holds for $k \leq p-1$ or $q-p+1\leq k\leq q-2$. In particular, the conjecture holds when $q=p$ is prime.
 Kaipa \cite{Kaipa17} improved this result and he proved that the deep hole conjecture holds for $k\geq \frac{q-1}2$.

A received word  $u$  is called ordinary if $d(u,\RS_{n,k})=n-\deg(u)$, that is, $u$ realizes the lower bound in Theorem \ref{thm1}.

Recall that for  a monic polynomial $f(x)\in \f [x]$ of degree $k+\ell$, $N(f(x), r)$ is  the number  of polynomials $g(x)\in\f[x]$  with  $\deg g(x) \leq k-1$ such that $f(x)+g(x)$ has exactly $r$ distinct roots in $\f$.
Then  a degree $k+\ell$ word $u$ is ordinary if and only if
\begin{align*}
N(u(x), k+\ell)>0.
\end{align*}

In this paper, we improve the results of Zhu and Wan \cite{ZhuWan12} about deep holes by allowing the information rate $k/q$ to be any constant in $(0,1)$. Our main results are the following.

\begin{thm}\label{thm2} Let $q$ be a power of a prime $p$ and $f\in \Mcal_{k+\ell}$.
\begin{itemize}
\item[(a)] Let $c=(k+\ell)/q$. We have $d(f,\RS_{q,k})=q-k-\ell$ provided that
\begin{align}
& ~~~~\frac{(p-1)c}{p}\ln \frac{1}{c}+(1-c)\ln \frac{1}{1-c}-\frac{1+c}{p}\ln(1+c)
\nonumber\\
&\ge (\ell-1)\left(\frac{\ln q}{q}+\frac{\ln(2p)}{\sqrt{q}}\right)+\frac{2}{3q}+\frac{3\ln q}{2q}.\label{eq:main1}
\end{align}
\item[(b)]
Let $c=(k+1)/q$. We have $d(f,\RS_{q,k})\le q-k-1$ provided that
\begin{align}
& ~~~~\frac{(p-1)c}{p}\ln \frac{1}{c}+(1-c)\ln \frac{1}{1-c}-\frac{1+c}{p}\ln(1+c)
\nonumber\\
&\ge (\ell-1)\left(\frac{\ln q}{2q}+\frac{\ln(2p)}{\sqrt{q}}\right)+\frac{2}{3q}+\frac{3\ln q}{2q}.\label{eq:main1}
\end{align}
\end{itemize}
\end{thm}

\begin{thm}\label{thm3} Let $q$ be a power of a prime $p$ and $f\in \Mcal_{k+\ell}$.
 \begin{itemize}
\item[(a)] For any constant $c\in (0,1)$, there are positive constants $p_0,q_0,\gamma_0$ (depending on $c$) such that 
$d(f,\RS_{q,k})= q-k-\ell$, provided that $p\ge p_0$, $q\ge q_0$, and $\ell\le 1+\gamma_0\sqrt{q}$.
\item[(b)] For each prime $p$, there are positive constants $q_0,\gamma_0$ (depending on $p$) such that \\
$d(f,\RS_{q,k})= q-k-\ell$,
provided that $3\le k+\ell\le 0.7q$, $q\ge q_0$, and $\ell\le 1+\gamma_0\sqrt{q}$.
\end{itemize}
 \end{thm}

More details about the numerical values of these parameters including the cases for
$p\in\{2, 3, 5,7\}$ can be found in Corollaries~1 and 2 in Section~4.

 \subsection{Preliminaries}

Our approach follows that of \cite{Gao21} using generating functions. The coefficients of the generating functions are from the group algebra generated by
equivalence classes consisting of polynomials with prescribed leading coefficients.

Throughout the paper, we shall use the  notations from \cite{Gao21}, which are summarized below.
\begin{itemize}
\item $\fq$ denotes the finite field with $q$ elements, where $q=p^s$ for some prime $p$, $\fqx=\fq\setminus \{0\}$.
\item $\fq[x]$ denotes the set of polynomials with coefficients in $\fq$.
\item For $f\in \fq[x]$, $[x^d]f$ denotes the coefficient of $x^d$ in $f$.
\item $\deg(f)$ denotes the degree of the polynomial $f$.
\item $\Mcal$ denotes the set of monic polynomials in $\fq[x]$, and $\Mcal_d$ denotes the set of polynomials of degree $d$ in $\Mcal$.
\item  Fix a positive integer $\ell$. Two polynomials $f,g \in \Mcal$ are {\em equivalent} if they have the same $\ell$ leading coefficients, that is
    \[
    \left[x^{\deg(f)-j}\right]f(x)=\left[x^{\deg(g)-j}\right]g(x),~~1\le j\le \ell.
    \]
This defines an equivalence relation, and we use $\langle f\rangle$ to denote the equivalence class represented by $f$. We also use $\Ecal$ to denote the set of all equivalence classes, and $\Mcal_d(\eps)$ denotes the set of polynomials in $\Mcal_d$ in the equivalence class $\eps$.
\end{itemize}

The remaining paper is organized as follows.
In Section~2 we recall the generating functions obtained in \cite{Gao21} for $N(f,r)$ for general $D\subseteq \fq$. Simple expressions are derived for the factorial moments of the distance between a given received word  and a random codeword in $\RS_{n,k}$.  In Section~3 we focus on the standard Reed-Solomon code $\RS_{q,k}$ and use the character sum bounds from \cite{LiWan20} to derive a sharper error bound for the factorial moments.
Proofs of our main results are given in Section~4.  Section~5 concludes our paper.

\section{Generating functions and moments of distance distribution over Reed-Solomon codes }

For $\eps\in \Ecal$, let $N_d(\eps,r)$ denote the number of polynomials in $\Mcal_d(\eps)$ which contain exactly $r$ distinct zeros in the give set $D$.
Define generating functions
\begin{align*}
F(z)&=\sum_{f\in \Mcal}\langle f\rangle z^{\deg(f)},\\
G(z,u)&=\sum_{\eps\in \Ecal}\sum_{d,r\ge 0}N_d(\eps,r)z^du^r.
\end{align*}

It is convenient to introduce the following element in the group algebra ${\mathbb C}\Ecal$:
\begin{align}
 E&=\frac{1}{q^{\ell}}\sum_{\eps\in \Ecal}\eps.\label{eq:E}
\end{align}
It is easy to see that (see, e.g., \cite{Gao21})
\begin{align}
 E\eps&=E,~~E^2=E,~~\forall \eps\in \Ecal.\label{eq:E2}
\end{align}
We shall use the following result \cite{Gao21}:
\begin{prop}\label{prop:Gzu}
Let $E$ be defined by \eqref{eq:E}. Then
\begin{align}
F(z)&=\sum_{d=0}^{\ell-1}\sum_{f\in \Mcal_d}\langle f\rangle z^d+\frac{(qz)^{\ell}}{1-qz}E,\label{eq:F}\\
G(z,u)&=F(z)\prod_{\alp\in D}(\langle 1\rangle+(u-1)z\langle x+\alp\rangle). \label{eq:G}
\end{align}
\end{prop}

Let $D_j$ denote the set of all $j$-subsets of $D$. The following result can be found in \cite{Gao21}, we include a short proof here for self completeness.
\begin{thm}\label{thm:thm1} The number of polynomials in $\Mcal_d(\eps)$ containing exactly $r$ linear factors is given by
\begin{align}\label{eq:NW}
N_d(\eps,r)&={n\choose r}q^{d-\ell-r}\sum_{j=0}^{d-\ell-r}{n-r\choose j}(-q)^{-j}+\sum_{j=d-\ell+1}^{d}{j\choose r}(-1)^{j-r}W_j(\eps),
\end{align}
where
\begin{align}
W_j(\eps)&=\sum_{\eta\in \Ecal_{d-j}}\sum_{S\in D_{j}}\left\llbracket \eta\prod_{\alp\in S}\langle x+\alp\rangle=\eps \right\rrbracket.
\end{align}
\end{thm}
\proof For $\eps\in \Ecal$, let $\left[z^d\eps\right]G(z,u)$ denote the coefficient of $z^d\eps$ in the generating function $G(z,u)$. Using \eqref{eq:E2} and \eqref{eq:G}, we obtain
\begin{align}
G(z,u)&=\frac{(qz)^{\ell}}{1-qz}(1+(u-1)z)^nE\nonumber\\
&~~~~+\left(\sum_{j=0}^{\ell-1}z^{j}\sum_{f\in \Mcal_j}\langle f\rangle\right)\prod_{\alp\in D}\left(\langle 1\rangle+(u-1)z\langle x+\alp\rangle\right),\nonumber\\
\left[z^d\eps\right]G(z,u)&=\left[z^{d-\ell} \right]\frac{1}{1-qz}\left( 1+(u-1)z\right)^n\nonumber\\
&~~~~+\sum_{j=d-\ell+1}^{d}(u-1)^{j}[\eps]\sum_{\eta\in \Ecal_{d-j}}\sum_{S\in D_{j}}\eta\prod_{\alp\in S}\langle x+\alp\rangle\nonumber\\
&=\sum_{j=0}^{d-\ell}{n\choose j}q^{d-\ell-j}(u-1)^{j}+\sum_{j=d-\ell+1}^{d}W_j(\eps)(u-1)^{j}.\label{eq:Gu}
\end{align}
It follows that
\begin{align}
N_d(\eps,r)&=\sum_{j=0}^{d-\ell}{n\choose j}q^{d-\ell-j}\left[u^r\right](u-1)^{j}+
\sum_{j=d-\ell+1}^{d}W_j(\eps)\left[u^r\right](u-1)^{j}\\
&=\sum_{j=0}^{d-\ell}{n\choose j}{j\choose r}q^{d-\ell-j}(-1)^{j-r}+
\sum_{j=d-\ell+1}^{d}W_j(\eps){j\choose r}(-1)^{j-r}.\label{eq:Ndr}
\end{align}
Changing the summation index $j:=j+r$ in the first sum in \eqref{eq:Ndr} and using
\begin{align}
{n\choose j}{j\choose r}&={n\choose r}{n-r\choose j-r}, \label{eq:bin}
\end{align}
we complete the proof. \qed

Let $\overline{D^j}$ denote the set of all $j$-tuples of {\em distinct elements} of $D$. We note
\begin{align}
W_j(\eps)&=\sum_{\eta\in \Ecal_{d-j}}\sum_{S\in D_{j}}\left\llbracket \eta\prod_{\alp\in S}\langle x+\alp\rangle=\eps \right\rrbracket \nonumber\\
&=\frac{1}{j!}\sum_{\eta\in \Ecal_{d-j}}\sum_{\vec{x}\in \overline{D^j}}\left\llbracket \eta\prod_{i=1}^{j} \langle x+x_i\rangle=\eps \right\rrbracket. \label{eq:Wj}
\end{align}

Theorem~\ref{thm:thm1} was used in \cite{Gao21} to obtain explicit expressions of $N_f(\eps,r)$ for $\ell\le 2$ and the first two moments of the distance distribution.


Let $\ex(X)$ denote the expected value of a random variable $X$. We shall use $n^{\underline m}$ to denote the falling factorial $n(n-1)\cdots (n-m+1)$. Our next result gives all the factorial moments expressed in terms of $W_j(\eps)$. This extends the corresponding results in \cite{Gao21} for the first two moments.
\begin{thm}
Let $Z:=Z(f,g)$ denote the distance between a  received word represented by  $f\in \Mcal_{k+\ell}$
 and a random codeword $g$ in $\RS_{n,k}$ (under uniform distribution, that is, each word in $\RS_{n,k}$ is chosen with probability $q^{-k}$).
 Let $Y:=n-Z$, and $\eps=\langle f\rangle$.  We have
\begin{align}\label{eq:Moments}
\Ex\left(Y^{\underline m}\right)&=\llbracket  m\le k\rrbracket n^{\underline m}q^{-m}
+\llbracket k+1\le m\le k+\ell\rrbracket q^{-k}m!W_m(\eps).
\end{align}
\end{thm}
\proof  Using \eqref{eq:Gu}, we obtain the following probability generating function of $Y$:
\begin{align*}
p_{k+\ell}(u)&=q^{-k}\left[z^n\eps\right]G(z,u)\\
&=\sum_{j=0}^{k} q^{-j}{n\choose j}(u-1)^j
+q^{-k}\sum_{j=k+1}^{k+\ell}W_j(\eps)(u-1)^{j}.
\end{align*}
Hence
\begin{align*}
\Ex\left(Y^{\underline m}\right)&=\frac{d^m}{(du)^m}p_{k+\ell}(u)\Bigr|_{u=1} \\
&=\llbracket m\le k\rrbracket m!{n\choose m}q^{-m}
+\llbracket k+1\le m\le k+\ell\rrbracket q^{-k}m!W_m(\eps).
\end{align*}
 \qed

\section{Standard Reed-Solomon code}

In this section we deal with standard Reed-Solomon code, that is, $D=\fq$ and hence $n=q$.
We also set $d=k+\ell$. In this case $W_j(\eps)$ can be estimated using Weil bounds on character sums. We first
rewrite \eqref{eq:Wj}  as
\begin{align}
W_j(\eps)
&=\frac{1}{j!}\sum_{\eta\in \Ecal_{k+\ell-j}}\sum_{\vec{x}\in \overline{\fq^j}}\left\llbracket \eta\prod_{i=1}^{j} \langle x+x_i\rangle=\eps \right\rrbracket. \label{eq:Wj1}
\end{align}

As in \cite{LiWan20}, $W_j(\eps)$ can be evaluated using the ''coordinate-sieve formula". The corresponding error terms can be expressed in terms of the following function, which appears in the exponential generating function of permutations with respect to two different types of cycle lengths. Let $S_m$ denote the set of all permutations of  $1,2,\ldots,m$. For $\tau\in S_m$, define $l(\tau)$ to be the total number of cycles of $\tau$ and $l'(\tau)$ to be the number of cycles of $\tau$ which are not multiples of $p$. The standard generating function argument gives \cite{FlaSed09,Gao21}
\begin{align*}
\sum_{m\ge 0}\frac{1}{m!}\sum_{\tau\in S_m}u^{l(\tau)}w^{l'(\tau)}z^m&=\exp\left(u\sum_{j\ge 1,p\mid j}\frac{z^j}{j}+uw\sum_{j\ge 1,p\nmid j}\frac{z^j}{j}\right)\\
&=(1-z)^{-uw}(1-z^p)^{-(u-uw)/p}.
\end{align*}
As in \cite{Gao21}, we define
\begin{align}
A_j(u,w)&= \frac{1}{j!}\sum_{\tau\in S_j}u^{l(\tau)}w^{l'(\tau)}\label{eq:cycle}\\
&=[z^j]\left((1-z)^{-uw}(1-z^p)^{-(u-uw)/p}\right)\label{eq:Agen}\\
&=\sum_{0\le i\le j/p}{uw+j-ip-1\choose j-ip}{(u-uw)/p+i-1 \choose i}.\label{eq:Aj}
\end{align}

Define
\begin{align}\label{eq:q1}
q_1=\min\{q,(\ell-1)\sqrt{q}\}.
\end{align}
Then we have the following estimate.
\begin{prop}\label{prop:Wj} Let $\eps\in \Ecal$. Then \begin{align}
\left|W_j(\eps)-{q\choose j}q^{k-j}\right|&\le \left(1-q^{-\ell}\right){\ell-1\choose \ell+k-j}q^{(\ell+k-j)/2}A_j(q,q_1/q).\label{eq:Wbound}
\end{align}
\end{prop}
\proof Let ${\hat \Ecal}$ denote the set of characters over the group $\Ecal$. Then by orthogonality of the characters, we have
\begin{align}
W_j(\eps)&=\frac{q^{-\ell}}{j!}\sum_{\eta\in \Ecal_{k+\ell-j}}\sum_{\vec{x}\in \overline{\fq^j}}
\sum_{\chi\in {\hat \Ecal}}\chi\left( \eps^{-1}\eta\prod_{i=1}^{j} \langle x+x_i\rangle\right)\nonumber\\
&={q\choose j}q^{k-j}+\frac{q^{-\ell}}{j!}\chi(\eps^{-1})\sum_{\chi\ne 1}\left(\sum_{\eta\in \Ecal_{k+\ell-j}}\chi(\eta)\right) \sum_{\vec{x}\in \overline{\fq^j}}
\chi\left( \prod_{i=1}^{j} \langle x+x_i\rangle\right).\label{eq:Wj2}
\end{align}

The sum in \eqref{eq:Wj1} can be estimated using Weil bound on character sums as shown in \cite{LiWan20}.
We note that $\sum_{\eta\in \Ecal_k}\chi(\eta)$  is the same as  $M_k(\chi)$ in \cite{LiWan20}, and hence
\begin{align}\label{eq:eta}
\left|\sum_{\eta\in \Ecal_{k+\ell-j}}\chi(\eta)\right|&\le {\ell-1\choose k+\ell-j}q^{(k+\ell-j)/2}.
\end{align}

The character sums involving the linear factors are the same as the corresponding ones given in \cite{LiWan20}, and the estimate of $G_{\tau}$ in \cite{LiWan20} gives
\begin{align}\label{eq:linearbound}
\left|\sum_{\vec{x}\in \overline{D^j}}
\chi\left( \prod_{i=1}^j \langle x+x_i\rangle\right)\right|
&\le \sum_{\tau\in S_j}q^{l-l'}q_1^{l'}=j!A_j(q,q_1/q).
\end{align}
Substituting \eqref{eq:eta} and \eqref{eq:linearbound} into \eqref{eq:Wj2}, we complete the proof. \qed

\medskip

Using Theorem~1 and Proposition~\ref{prop:Wj},  we immediately obtain the following.
\begin{thm}\label{thm:thm6} Assume $D=\fq$. The number $N_{k+\ell}(\eps,r)$ of monic polynomials of degree $k+\ell$ containing exactly $r$ zeros in $\fq$ with leading coefficients $\eps$ satisfies
\begin{align}
&~~\left|N_{k+\ell}(\eps,r)-{q\choose r}q^{k-r}\sum_{j=0}^{k+\ell-r}{q-r\choose j}(-q)^{-j}\right|\nonumber\\
&\le \left(1-q^{-\ell}\right) \sum_{j=k+1}^{k+\ell}{j\choose r}{\ell-1\choose k+\ell-j}q^{(k+\ell-j)/2}A_j(q,q_1/q).\label{eq:thm2}
\end{align}
\end{thm}

We note that this is essentially \cite[Theorem~1.5]{LiWan20} except that our $A_j(q,q_1/q)$ replaces their factor
\begin{align}
{q/p+q_1+j-1 \choose j}. \label{eq:Wan}
\end{align}

\section{Proof of the main results}
 The following lemma provides a much sharper upper bound for $A_j(q,q_1/q)$ than the binomial number in \eqref{eq:Wan}. It will be used to prove our main results.
\begin{lemma}\label{lem:Abound}
Let $q_1=\min\{q,(\ell-1)\sqrt{q}\}$, $\gamma=q_1/q$ and $c:=j/q$. We have
\begin{align}
\ln A_j(q,q_1/q)&\le \frac{j}{p}\ln \frac{q+j}{j}+\frac{q-q_1}{p}\ln\left(\frac{q+j}{q}\right)
-q_1\ln\left(1- \left(\frac{j}{q+j}\right)^{1/p}\right),\label{eq:general}\\
&\le q\left(\frac{c}{p}\ln \frac{1+c}{c}+\frac{1-\gamma}{p}\ln(1+c)
+\gamma \ln (2p)\right).\label{eq:large}
\end{align}
\end{lemma}
\proof From \eqref{eq:Agen}, we have
\begin{align*}
\sum_{ j\ge 0}A_j(q,q_1/q)z^j&=(1-z)^{-q_1}(1-z^p)^{-(q-q_1)/p}.
\end{align*}
By \eqref{eq:cycle}, we have
\[
A_j(q,q_1/q)\ge 0,~~~(j\ge 0).
\]
It follows that, for $0<y<1$,
\begin{align}
A_j(q,q_1/q)&\le y^{-j}(1-y)^{-q_1}(1-y^p)^{-(q-q_1)/p}\nonumber\\
\ln A_j(q,q_1/q)&\le -j\ln y-q_1\ln(1-y)-\frac{q-q_1}{p}\ln\left(1-y^p\right). \label{eq:Abound1}
\end{align}
To minimize the above upper bound, we choose $y$  near the solution to the following saddle point equation
\begin{align}\label{eq:saddle}
-\frac{j}{y}+q_1\frac{1}{1-y}+(q-q_1)\frac{y^{p-1}}{1-y^p}=0,~~
i.e.,~~~q_1\frac{y}{1-y}+(q-q_1)\frac{y^p}{1-y^p}=j.
\end{align}

 When $q_1$ is much smaller than $q$, we may  choose
\begin{align*}
y=\left(\frac{j}{q+j}\right)^{1/p}
 \end{align*}
 to be an approximate solution to \eqref{eq:saddle}.
Substituting this into \eqref{eq:Abound1}, we obtain  \eqref{eq:general}.

 Noting
\begin{align}\label{eq:exp1}
e^{-t}&\le 1-\frac{3t}{4}, &(0\le t\le 0.6)
\end{align}
we obtain
\begin{align}
\left(\frac{j}{q+j}\right)^{1/p}&=\exp\left(\frac{-1}{p}\ln\frac{q+j}{j}\right)\le \exp\left(\frac{-\ln 2}{p}\right)\le 1-\frac{3\ln 2}{4p},\nonumber\\
-\ln\left(1- \left(\frac{j}{q+j}\right)^{1/p}\right)&\le -\ln\frac{3\ln 2}{4p}
\le \ln (2p). \label{eq:bound}
\end{align}
Substituting \eqref{eq:bound} into \eqref{eq:general}, we obtain \eqref{eq:large}. \qed

\medskip

\noindent {\bf Remark} Bound \eqref{eq:large} is sufficient for our purpose. One can derive sharper bounds depending on the ranges of $p,q,\gamma,c$. For example, when $p=2$,
\begin{align*}
y(\gamma,c)&=\frac{\sqrt{\gamma^2+4c^2+4c}-\gamma}{2(1+c)}
\end{align*}
is the exact solution to the saddle point equation, which can be used to obtain a sharper estimate for $A_j(q,q_1/q)$ than \eqref{eq:general}.
 When $\gamma$ is close to 1, a different choice of $y$ gives better estimate for $A_j(q,q_1/q)$ than
\eqref{eq:general}.

\medskip

The following inequality \cite[(5)]{Bo85} will be used to estimate binomial numbers.
\begin{align}\label{eq:bin}
{n\choose m}&\ge  e^{-1/6}
  \left(\frac{n}{2\pi m(n-m)}\right)^{1/2}\left(\frac{n}{m}\right)^{m}\left(\frac{n}{n-m}\right)^{n-m}.
  &(0< m<n)
\end{align}

\noindent {\bf Remark}: We may use \eqref{eq:bin} to show
 \[
 \ln {q/p+q_1+j-1 \choose j} -\ln A_j(q,q_1/q)\to \infty
 \]
 when $j=cq$ and $q_1=o(q)$.

If $p=q$ then it follows from \eqref{eq:Aj} that
\[
\ln {q/p+q_1+j-1 \choose j} -\ln A_j(q,q_1/q)=\ln \frac{{q_1+j\choose j}}{{q_1+j-1\choose j}}
=\ln \frac{q_1+j}{q_1}\to \infty.
\]

When $q\ge p^2$,  we apply \eqref{eq:bin} to obtain
\begin{align}\label{eq:Wan1}
\ln {q/p+q_1+j-1 \choose j}&\ge cq\ln \frac{(c+1/p)q+q_1-1}{cq}+(q/p+q_1-1)\ln \frac{(c+1/p)q+q_1-1}{q/p+q_1-1}\nonumber\\
&~~~-\frac{1}{2}\ln \frac{2\pi cq}{1+cp}-\frac{1}{6}.
\end{align}
The RHS of \eqref{eq:Wan1} is asymptotically larger than
\[
cq\ln \frac{1+cp}{cp}+\frac{q}{p}\ln(1+cp).
\]
It follows from \eqref{eq:large} that, if $q\ge p^2$, then
\[
\ln {q/p+q_1+j-1 \choose j}-\ln A_j(q,q_1/q)>\frac{q}{p}\ln\frac{1+cp}{1+c}\to \infty.
\]

\begin{thm}\label{thm:main} Let $q$ be a power of a prime $p$,  $\gamma=(\ell-1)q^{-1/2}\le 1$, and $f\in \Mcal_{k+\ell}$.
\begin{itemize}
\item[(a)] Let $c=(k+\ell)/q$. We have $d(f,\RS_{q,k})=q-k-\ell$ provided that
\begin{align}
& ~~~~\frac{(p-1)c}{p}\ln \frac{1}{c}+(1-c)\ln \frac{1}{1-c}-\frac{1+c}{p}\ln(1+c)
-\left(\frac{1}{6q}+\frac{\ln q}{q}+\frac{1}{2q}\ln(2q\pi c(1-c))\right)
\nonumber\\
&\ge \gamma\left(\frac{\ln q}{\sqrt{q}}+\ln(2p)-\frac{1}{p}\ln(1+c)\right).\label{eq:main1}
\end{align}
\item[(b)]
Let $c=(k+1)/q$. We have $d(f,\RS_{q,k})\le q-k-1$ provided that
\begin{align}
& ~~~~\frac{(p-1)c}{p}\ln \frac{1}{c}+(1-c)\ln \frac{1}{1-c}-\frac{1+c}{p}\ln(1+c)
-\left(\frac{1}{6q}+\frac{\ln q}{q}+\frac{1}{2q}\ln(2q\pi c(1-c))\right)
\nonumber\\
&\ge \gamma\left(\frac{\ln q}{2\sqrt{q}}+\ln(2p)-\frac{1}{p}\ln(1+c)\right).\label{eq:main2}
\end{align}
\end{itemize}
 \end{thm}
 \begin{proof}
For simplicity, define
\begin{align*}
f(p,c)&=\frac{(p-1)c}{p}\ln \frac{1}{c}+(1-c)\ln \frac{1}{1-c}-\frac{1+c}{p}\ln (1+c),\\
g(q,c)&=\frac{1}{6q}+\frac{\ln q}{q}+\frac{1}{2q}\ln \left(2q\pi c(1-c)\right),\\
h_1(p,q,c)&=\frac{\ln q}{\sqrt{q}}+ \ln (2p)-\frac{1}{p}\ln (1+c),\\
h_2(p,q,c)&=\frac{\ln q}{2\sqrt{q}}+ \ln (2p)-\frac{1}{p}\ln (1+c).
\end{align*}

For $k+1\le r\le k+\ell$, we first note
\begin{align*}
r!\prob(Y\ge r)\le \ex\left(Y^{\underline {r}}\right)\le (k+\ell)^{\underline {r}}\,\prob(Y\ge r),
\end{align*}
and hence $\prob(Y\ge r)>0$ if and only if $\ex\left(Y^{\underline {r}}\right)>0$.
It follows from \eqref{eq:Moments} and \eqref{eq:Wbound} that $\prob(Y\ge r)>0$ if
\begin{align}
{q\choose r}\ge {\ell-1\choose \ell+k-r}q^{(r+\ell-k)/2}A_r(q,q_1/q). \label{eq:Pbound}
\end{align}

For $c=r/q$, we use \eqref{eq:bin}  to obtain
\begin{align}
\ln {q\choose r}&\ge q\left(c\ln \frac{1}{c}+(1-c)\ln \frac{1}{1-c}\right)-\frac{1}{2}\ln \left(2q\pi c(1-c)\right)-\frac{1}{6}.\label{eq:bin1}
\end{align}

\noindent  For part~(a), we substitute $r=k+\ell$ into \eqref{eq:Pbound} to obtain
\begin{align*}
{q\choose k+\ell}&\ge q^{\ell}A_{k+\ell}.
\end{align*}
Taking the logarithm on both sides,  using \eqref{eq:large} and \eqref{eq:bin1}, and dividing by $q$, we obtain the following sufficient condition for $\prob(Y\ge k+\ell)>0$:
\begin{align*}
&~~c\ln \frac{1}{c}+(1-c)\ln \frac{1}{1-c}-\frac{1}{2q}\ln \left(2q\pi c(1-c)\right)-\frac{1}{6q}\\
&\ge \frac{\ell \ln q}{q}+
\frac{c}{p}\ln \frac{1+c}{c}+\frac{1-\gamma}{p}\ln(1+c)+\gamma \ln (2p).
\end{align*}
Now \eqref{eq:main1} follows by noting $\prob(Y=k+\ell)=\prob(Y\ge k+\ell)$ and $\ell=1+\gamma\sqrt{q}$.

\noindent  For part~(b), we substitute $r=k+1$ into \eqref{eq:Pbound} to obtain
\begin{align*}
{q\choose k+1}&\ge q^{(\ell+1)/2}A_{k+1}.
\end{align*}
Now \eqref{eq:main2} follows from the same argument as in part~(a) with $h_2$ replacing $h_1$.\qed

The following figure shows the intervals where $f(p,c)\ge 0$ for $p\in \{2,3,5,7,17\}$.
\begin{center}
\begin{figure}[h]
\includegraphics[scale=0.5]{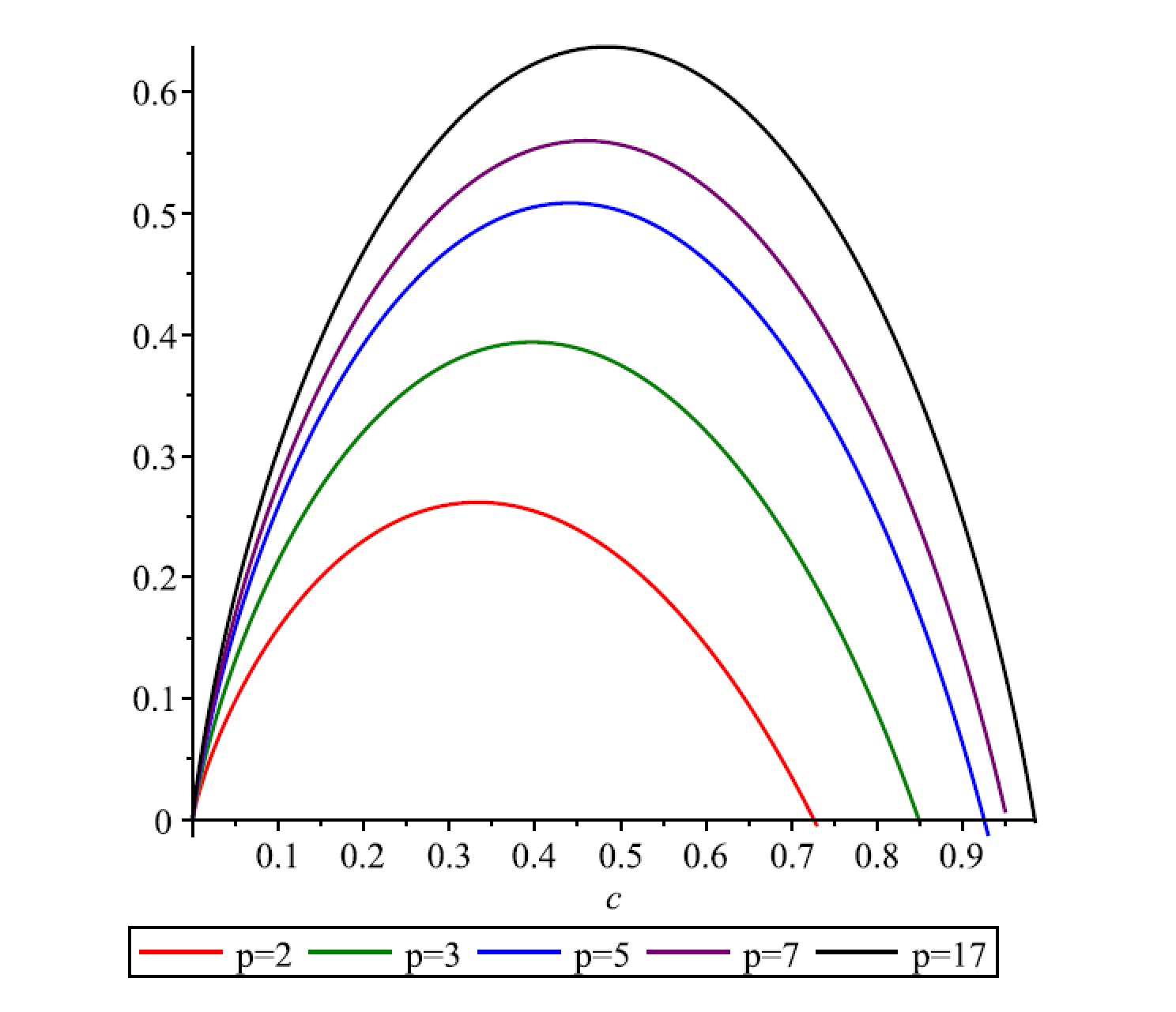}
\end{figure}
\end{center}
The next two corollaries provide some precise ranges of $c$ for $p\le 7$.  Since the deep hole conjecture has been verified for $k\ge \lfloor (q-1)/2\rfloor$, we focus on $k\le (q-2)/2$ in Corollary~\ref{cor:cor2}.
\begin{cor}\label{cor:cor2} Let $q$ be a power of a prime $p$, $c=(k+1)/q$ and $\gamma=(\ell-1)/\sqrt{q}$.
Let $f(p,c),g(q,c),h_2(p,q,c)$ be defined in the proof of Theorem~\ref{thm:main}.
Then  $\displaystyle d(f,\RS_{q,k})\le q-k-1$ for each of the following cases.
\begin{align*}
(a)& \quad p=2, q\ge 2^5, 3\le k\le \frac{q-2}{2},  &\gamma\le \frac{f(p,c)-g(q,c)}{h_2(p,q,c)},\\
(b)& \quad p=3, q\ge 3^3, 1\le k\le \frac{q-2}{2},  &\gamma\le \frac{f(p,c)-g(q,c)}{h_2(p,q,c)},\\
(c)& \quad p=5, q\ge 5^2, 1\le k\le \frac{q-2}{2},  &\gamma\le \frac{f(p,c)-g(q,c)}{h_2(p,q,c)},\\
(d)& \quad p\ge 7, q\ge p, 1\le k\le \frac{q-2}{2}, &\gamma\le \frac{f(p,c)-g(q,c)}{h_2(p,q,c)}.
\end{align*}
\end{cor}

\proof  We first note
\begin{align*}
\frac{\partial^2}{(\partial c)^2}(f(p,c)-g(q,c))&=\frac{1+c}{c(c-1)}+\frac{1}{cp}+\frac{1-2c+2c^2}{2qc^2(1-c)^2}\\
&\le \frac{1+c}{c(c-1)}+\frac{1}{2c}+\frac{1-2c+2c^2}{2c(1-c)^2}&(\hbox{using }p\ge 2,q\ge 1/c)\\
&=\frac{2c^2+3c-3}{2(1+c)(1-c)^2}<0. &(\hbox{if }0\le c\le 0.68).
\end{align*}
Hence for given $p,q$, each $f(p,c)-g(q,c)$ is concave down on $[0,0.5]$. Thus we only need to verify
$f(p,c)-g(q,c) >0$ at the end points of the respective interval.

\noindent For part~(a) we have
\begin{align*}
f(2,0.5)-g(q,0.5)&>0.041,~ f(2,4/q)-g(q,4/q)>0.0187.  &(q\ge 32)
\end{align*}

\noindent For part~(b) we have
\begin{align*}
f(3,0.5)-g(q,0.5)&>0.1772,~f(3,2/q)-g(q,2/q)>0.0005.  &(q\ge 27)
\end{align*}

\noindent For part~(c) we have
\begin{align*}
f(5,0.5)-g(q,0.5)&>0.2933,~f(5,2/q)-g(q,2/q)>0.0373.  &(q\ge 25)
\end{align*}

\noindent For part~(d) we have $f(p,c)>f(7,c)$ if $p>7$, and
\begin{align*}
f(7,0.5)-g(q,0.5)&>0.0837,~f(7,2/q)-g(q,2/q)>0.0424.  &(q\ge 7)
\end{align*}
\qed

\begin{cor}\label{cor:cor3} Let $q$ be a power of a prime $p$, $c=(k+\ell)/q$, $\gamma$,  $f(p,c)$, $g(q,c)$ and $h_1(p,q,c)$ be as defined in the proof of Theorem~\ref{thm:main}.
Then  $\displaystyle d(f,\RS_{q,k})=q-k-\ell$ for each of the following cases.
\begin{align*}
(a)& ~~ p=2, q\ge 2^8, \frac{3}{q}\le c\le 0.7, &\gamma\le \frac{f(p,c)-g(q,1/2)}{h_1(p,q,c)},\\
(b)& ~~ p=3, q\ge 3^4, \frac{3}{q}\le c\le 0.8, &\gamma\le \frac{f(p,c)-g(q,1/2)}{h_1(p,q,c)},\\
(c)& ~~ p=5, q\ge 5^3, \frac{2}{q}\le c\le 0.9, &\gamma\le \frac{f(p,c)-g(q,1/2)}{h_1(p,q,c)},\\
(d)& ~~ p\ge 7, q\ge 7^4, \frac{2}{q}\le c\le 0.95,&\gamma\le \frac{f(p,c)-g(q,1/2)}{h_1(p,q,c)}.
\end{align*}
\end{cor}

\proof  The proof is similar to that of Corollary~\ref{cor:cor2}. To maximize the range of $c$, we use the bound $g(q,c)\le g(q,1/2)$, and note that $f(p,c)$ is concave up on $(0,1)$ for each $p\ge 2$.
Hence for given $p,q$ we only need to verify
$f(p,c)-g(q,1/2) >0$ at the end points of the respective interval.

\noindent For part~(a) we have
\begin{align*}
f(2,0.7)-g(q,0.5)&>0.0009,~ f(2,3/q)-g(q,1/2)>0.0069.  &(q\ge 2^8)
\end{align*}

\noindent For part~(b) we have
\begin{align*}
f(3,0.8)-g(q,1/2)&>0.002,~f(3,3/q)-g(q,1/2)>0.0189.  &(q\ge 3^4)
\end{align*}

\noindent For part~(c) we have
\begin{align*}
f(5,0.5)-g(q,1/2)&>0.0011,~f(5,2/q)-g(q,1/2)>0.0044.  &(q\ge 5^3)
\end{align*}

\noindent For part~(d) we have $f(p,c)>f(7,c)$ if $p>7$, and
\begin{align*}
f(7,0.95)-g(q,1/2)&>0.0004,~f(7,2/q)-g(q,1/2)>0.0007.  &(q\ge 7^4)
\end{align*}

\medskip

\noindent {\em Proof of Theorem~\ref{thm2} and Theorem~\ref{thm3} }: Theorem~2 follows immediately from Theorem~7 by noting
\[
g(q,c)\le g(q,1/2)\le \frac{2}{3q}+\frac{3\ln q}{2q}.
\]
For Theorem~3(a), we may set $p_0=\frac{1+c}{1-c}$. Then, for $p\ge p_0$, we have
\begin{align*}
f(p,c)&\ge\frac{c}{2}\ln \frac{1}{c}+\left(1-c-\frac{1+c}{p}\right)\ln (1+c)\ge\frac{c}{2}\ln \frac{1}{c}.
\end{align*}
Choosing a sufficiently large positive $q_0$ such that 
\[
\frac{2}{3q_0}+\frac{3\ln q_0}{2q_0}<\frac{c}{2}\ln \frac{1}{c},
\] and
setting  $\gamma_0=(f(p_0,c)-g(q_0,1/2))/h_1(p_0,q_0,0)$, we obtain Theorem~3(a). \\
Theorem~\ref{thm3}(b) follows from Corollary~2 by setting $q_0=7^4$ and $\gamma_0=(f(p_0,c)-g(q_0,1/2))/h_1(p_0,q_0,0)$. \qed\\
\end{proof}

\section{Conclusion}

In this paper we used the generating function approach to derive simple expressions of the factorial moments of the distance distribution over $\RS_{n,k}$.
This recovers the Li-Wan's counting formula. We also  apply the saddle point estimate
to derive a sharper bound on the error term. This improves
the previous estimate and leads to several new applications on deep hole problem and
ordinary words problem.

\medskip

\noindent {\bf Acknowledgement} We would like to thank Prof. D. Wan for his helpful comments which  improve the presentation of the paper.


\begin{thebibliography}{999}

\bibitem{Bo85} B.~Bollob\'{a}s,
Random Graphs, Academic Press, 1985.

\bibitem{Cafure12} A. Cafure, G. Matera, and M. Privitelli. Singularities of symmetric hypersurfaces and an application to reed-solomon codes,
    {\em Advances in
Mathematics of Communications} {\bf 6}, 69�C94, 2012.


\bibitem{ChenMu07} Qi Cheng and Elizabeth Murray, On deciding deep holes of Reed-Solomon codes, In {\em Theory and applications of models of computation},
volume {\bf 4484} of {\em Lecture Notes in Comput. Sci.},  296-305. Springer,
Berlin, 2007.

\bibitem{CW04} Q. Cheng and D. Wan, \emph{On the list and Bounded distance Decodibility
of Reed-Solomon Codes}, FOCS (2004), 335-341.

\bibitem{CW07} Q. Cheng and D. Wan,
  \emph{On the list
and bounded distance decodability of Reed-Solomon codes},
 SIAM J. Comput. \textbf{37} (2007), no. 1, 195-209.



\bibitem{FlaSed09}  P. Flajolet and R. Sedgewick,  Analytic Combinatorics, Cambridge University Press, 2009.

\bibitem{Gao21} Z. Gao, Counting polynomials over finite fields with
prescribed leading coefficients and linear factors,   arXiv:2105.12845v2.



\bibitem{Kaipa17} K. Kaipa, Deep holes and MDS extensions of Reed-Solomon codes,  {\em IEEE Transactions on Information Theory} {\bf 99}, 4940--4948, 2017.








\bibitem{LiWan20} J.  Li and  D.  Wan,  Distance  distribution  in  Reed-Solomon  codes,  {\em IEEE Trans. Inform. Theory} {\bf 66}
(2020), no. 5, 2743--2750.


\bibitem{LiYWan08} Y. Li and D. Wan,  On error distance of Reed-Solomon codes,
{\em Science in China Series A: Mathematics} {\bf 51}, 1982--1988, 2008.

\bibitem{Liao11} Q. Liao, On Reed-Solomon codes, {\em Chinese Annals of Mathematics,
Series B} {\bf 32}, 89--98, 2011.



\bibitem{WuHong12}  R. Wu and S. Hong, On deep holes of standard Reed-Solomon codes,��\ {\em Science China Mathematics} {\bf  55}, 2447--2455, 2012.

\bibitem{ZhangFuLiao13} J. Zhang, F.-W. Fu, and Q. Liao, New deep holes of generalized Reed-Solomon codes, {\em Scientia Sinica} {\bf 43}, 727--740, 2013.


\bibitem{ZhuWan12} G. Zhu and D. Wan, Computing error distance of Reed-Solomon codes.
{\em TAMC 2012, LNCS {\bf 7287}}, 214--224, 2012.


\end{thebibliography}
\end{document}